\RequirePackage{fix-cm}
\documentclass[smallextended]{svjour3}       
\smartqed  
\usepackage{graphicx}
\usepackage{epstopdf}
\usepackage{geometry}
\usepackage{graphicx,amssymb,amstext,amsmath}
\usepackage[justification=centering]{caption}

\usepackage{epsfig}
\usepackage{subfigure}
\usepackage{amsmath}
\usepackage{algorithmic}
\usepackage{graphics}
\usepackage{amssymb}
\usepackage{alltt}
\usepackage{float}
\usepackage{array}
\usepackage{xspace}
\usepackage{footmisc}
\usepackage{comment}
\usepackage{algorithm}
\usepackage{algorithmic}
\usepackage[misc]{ifsym}

\newcounter{RomanNumber}

\begin{document}

\title{A Filter of Minhash for Image Similarity Measures
}
\subtitle{}


\author{Jun Long $^\dagger$$^{\natural}$        \and
        Qunfeng Liu $^\dagger$$^{\natural}$        \and
        Xinpan Yuan $^\ddagger$$^{\natural}$  \and
        Chengyuan Zhang $^\dagger$$^{\natural}$     \and
        Junfeng Liu  $^\dagger$$^{\natural}$     \and
}


\institute{Jun Long \at
              \email{jlong@csu.edu.cn}           
            \and
            Qunfeng Liu\at
              \email{qunfengliu@csu.edu.cn}
             \and
            \Letter XinPan Yuan \at
              \email{xpyuan@hut.edu.cn}
              \and
            Chengyuan Zhang \at
              \email{cyzhang@csu.edu.cn}           
              \and
              Junfeng Liu\at
              \email{junfengliu@csu.edu.cn}
              \and
           $^\dagger$ School of Information Science and Engeering, Central South University, PR China\\
           $^{\natural}$ Big Data and Knowledge Engineering Institute, Central South University, PR China\\
$^{\ddagger}$ School of Computer, Hunan University of Technology, China
}

\date{Received: date / Accepted: date}

\maketitle

\begin{abstract}
Image similarity measures play an important role in nearest neighbor search and duplicate detection for large-scale image datasets. Recently, Minwise Hashing (or Minhash) and its related hashing algorithms have achieved great performances in large-scale image retrieval systems. However, there are a large number of comparisons for image pairs in these applications, which may spend a lot of computation time and affect the performance. In order to quickly obtain the pairwise images that theirs similarities are higher than the specific threshold T (e.g., 0.5), we propose a dynamic threshold filter of Minwise Hashing for image similarity measures. It greatly reduces the calculation time by terminating the unnecessary comparisons in advance. We also find that the filter can be extended to other hashing algorithms, on when the estimator satisfies the binomial distribution, such as b-Bit Minwise Hashing, One Permutation Hashing, etc. In this pager, we use the Bag-of-Visual-Words (BoVW) model based on the Scale Invariant Feature Transform (SIFT) to represent the image features. We have proved that the filter is correct and effective through the experiment on real image datasets.
\end{abstract}

\section{Introduction}
\label{intro}

In recent years, with the rapid development of the Mobile Internet and social multimedia, a large number of images and videos are generated in the Internet every day. In the Mobile Internet era, people can take all kinds of pictures at any time or in any place and share them with their friends on the Internet, which results in the explosive growth of digital pictures. At present, hundreds and millions of pictures are uploaded to social media platforms every day, such as Facebook, Twitter, Flickr and so on. How to quickly search the similar images becomes a hot topic for the multimedia researchers, and the related research areas are also concerned ~\cite{DBLP:conf/cvpr/JegouDSP10}.

Image similarity measures aim to estimate whether a given pair of images is similar or not. It plays an important role in nearest neighbor search and near-duplicate detection for large-scale image resources. Recently, the Bag-of-Visual-Words (BoVW) model ~\cite{DBLP:journals/tmm/ZhengWLT15,DBLP:conf/cvpr/ZhengWZT14}, with local features, such as SIFT ~\cite{DBLP:journals/ijcv/Lowe04}, has been proven to be the most successful and popular local image descriptors. In the BoVW model, a “bag” of visual words is used to represent each image. The visual words are usually generated by clustering the extracted SIFT features. The SIFT descriptor is widely used in image matching ~\cite{DBLP:journals/corr/abs-1710-02726,DBLP:conf/mm/WangLWZ15,DBLP:journals/tip/WangLWZ17} and image search ~\cite{DBLP:journals/tomccap/ZhouLLT13,DBLP:journals/tip/LiuLZZT14a,DBLP:journals/tip/WangLWZZH15}.

At the beginning, the Minwise Hashing was mainly designed for measuring the set similarity. The algorithm is widely used for near-duplicate web page detection and clustering ~\cite{DBLP:journals/cn/BroderGMZ97,DBLP:conf/sigir/Henzinger06,DBLP:journals/tip/WangLWZZH15}, set similarity measures ~\cite{DBLP:conf/www/BayardoMS07}, nearest neighbor search ~\cite{DBLP:conf/stoc/IndykM98}, large-scale learning ~\cite{DBLP:conf/nips/LiSMK11,DBLP:conf/cikm/WangLZ13}, etc. In recent years, Minwise Hashing has been applied to the computer vision applications. Weighted min-Hash method has been proposed to find the near duplicated images. Grauman ~\cite{DBLP:conf/cvpr/JainKG08} combined the distance metric learning with the min-Hash algorithm to improve the image retrieval performance. A new method of highly efficient min-Hash generation for image collections is proposed by Chum and Matas ~\cite{DBLP:conf/cvpr/ChumM12}. Zhao developed an efficient matching technique for linking large image collections namely Sim-Min-Hash ~\cite{DBLP:conf/mm/ZhaoJG13}.Qu ~\cite{DBLP:conf/icimcs/QuSYL13} proposed a spatial min-Hash algorithm for similar image searching. In addition, some multimedia researchers proposed learning hash function for image similarity search ~\cite{DBLP:conf/sigir/WangLWZZ15,DBLP:journals/ivc/WuW17,DBLP:journals/cviu/WuWGHL18,DBLP:conf/pakdd/WangLZW14}. All these methods have achieved great performances in image similarity measures and image searching. However, in many image retrieval or search systems, there are huge amounts of comparisons for image pairs, which may spend a lot of computation time and have a negative impact on the performance of the systems.

Inspired by the successes of Minwise Hashing in image similarity measures, our main contributions are as follows:

\begin{itemize}
\item  We propose a dynamic threshold filter of Minwise Hashing for image similarity measures. The filter divides the whole fingerprint comparison process into a series of comparison points and sets the corresponding thresholds. At the k-th comparison point, the method will filter out the pairwise images whose similarities are less than the lower threshold TL(k) in advance. Meanwhile, the algorithm will output the pairwise images when theirs similarities are higher than the upper threshold TU(k). It greatly reduces the calculation time by terminating the unnecessary comparisons in advance.
\item  We find that the filter can be extended to other hashing algorithms for image similarity measures, as long as the estimator satisfies the binomial distribution, such as b-Bit Minwise Hashing, One Permutation Hashing, etc.
\end{itemize}

\noindent\textbf{Roadmap.}
The rest of the paper is organized as follows: Section 2 discusses the related works. Section 3 describes the dynamic threshold filter in detail. In Section 4, the filter is experimentally verified on real image databases. Section 5 gives conclusions.

\subsection{Image Representation and Similarity Measures}
\label{relwork}

This section reviews the BoVW model based on the SIFT to represent the image features, as well as the Minwise Hashing for measuring the set similarity.
In this paper, the relevant notations are shown in Table~\ref{tab:nnd notation}.

\begin{table}
	\centering
    \small
    \caption{Notations} \label{tab:nnd notation}
	\begin{tabular}{|p{0.27\columnwidth}| p{0.62\columnwidth} |}
		\hline
		~R & resemblance \\ \hline\hline
		~S & set                                \\ \hline
        ~$J(S_1,S_2)$   & the original Jaccard similarity of $S_1$ and $S_2$                                \\ \hline
	  	~K   & the k-th comparison point                                \\ \hline
	  	~$\pi$   & a random permutation function                                \\ \hline
	  	~$\Omega$   & represent the whole elements in the process of random permutation  \\ \hline
	  	~$\pi(S)$   & the hash values of the set S when given a hash function $\pi(.)$ \\ \hline
	  	~$min(\pi(S))$   & the minimum hash value in $\pi(S)$                              \\ \hline
	  	~Pr   & probability                                \\ \hline
	  	$R_M$   & the estimator of Minwise Hashing                               \\ \hline
	  	$R_{M}(k)$  &  the estimated similarity of Minwise Hashing at the $k$-th comparison point                               \\ \hline
	  	T   &  the number of times that the fingerprints are equal     \\ \hline
        $T$   &  a specified threshold    \\ \hline
        $T_{L}(k)$ &  the lower threshold at the $k$-th comparison point      \\ \hline
        $T_{U}(k)$   & the upper threshold at the $k$-th comparison point                                \\ \hline
        X  & the number of times that the fingerprints are equal                                \\ \hline
         F(x)& the distribution function of X                                \\ \hline
       e & a small probability                               \\ \hline
       m & a solution of the equation                               \\ \hline
              H & the hypothesis testing                              \\ \hline
	\end{tabular}
    \vspace{-4mm}
\end{table}

\subsubsection{Image Representation}
Scale Invariant Feature Transform (SIFT) ~\cite{DBLP:journals/ijcv/Lowe04} is a computer vision algorithm, which detects and describes local features in images. The SIFT feature descriptor is based on the appearance of the object at particular points. Besides, the descriptor is invariant to uniform scaling, orientation, illumination changes, and partially invariant to affine distortion. In the Bag-of-Visual-Words model, a set of visual words V ~\cite{DBLP:conf/iccv/SivicZ03,KAISWang17,NeuroWang13,DBLP:journals/corr/abs-1708-02288,NNLS2018} is constructed by the SIFT descriptor. The method builds the vocabulary through the K-means clustering algorithm from the training image datasets ~\cite{DBLP:conf/cvpr/NisterS06,DBLP:conf/cvpr/PhilbinCISZ07,DBLP:conf/ijcai/WangZWLFP16,LinMMM13,WangMM13,Wangarxiv2018,DBLP:journals/tnn/WangZWLZ17}. For each image, the SIFT features are assigned to the nearest cluster center and give the visual word representation.

For a vocabulary V, each visual word is encoded with the unique identifier from {1, … ,$\mid$V$\mid$}, where the variable $\mid$V$\mid$ is defines as the size of the vocabulary V. A set $S_i$ of words $S_i \subseteq $ V is a local representation, which does not store the number of features but only focusing on whether they present or not.

Therefore, each image can be represented by a visual word set S. The similarity of the pairwise images is equivalent to measure the similarity of visual sets ~\cite{DBLP:conf/bmvc/ChumPZ08,DBLP:conf/mm/WangLWZZ14}. We assume that $S_1$ and $S_2$ are the visual sets from a pair of images. So the similarity between two images can be defined as the Jaccard coefficient:
\begin{equation}
\mathcal{R}= sim(S_1, S_2) = \dfrac{S_1 \cap S_2}{S_1 \cup S_2}=Jaccard(S_1, S_2)
\end{equation}

\subsubsection{Minwise Hashing}
Minwise Hashing (or Minhash) is a Locality Sensitive Hashing, and is considered to be the most popular similarity estimation methods. It keeps a sketch of the data and provides an unbiased estimate of pairwise Jaccard similarity. In 1997, Andrei Broder and his colleagues invented the Minwise Hashing algorithm for near-duplicate web page detection and clustering ~\cite{DBLP:journals/cn/BroderGMZ97}. Recently, the algorithm is widely used in many applications including duplicate detection ~\cite{DBLP:conf/sigir/Henzinger06}, all-pairs similarity ~\cite{DBLP:conf/www/BayardoMS07}, nearest neighbor search ~\cite{DBLP:conf/stoc/IndykM98}, large-scale learning ~\cite{DBLP:conf/nips/LiSMK11} and computer visions ~\cite{DBLP:journals/pr/WuWGL18,TC2018,DBLP:journals/pr/WuWLG18}.
According to the process of Minwise Hashing, the algorithm requires K (commonly, K=1000) independent random permutations to deal with the datasets.It denotes π as a random permutation function: $\pi$: $\Omega \leftarrow\Omega$ and $min(\pi(S))$=$min_{i\in S}(\pi(i))$ . The similarity between two non-empty sets $S_1$ and $S_2$ is defined as:
\begin{equation}
Pr(min(\psi(\mathcal{S}_1))= min(\psi(\mathcal{S}_2))) = \dfrac{\mathcal{S}_1 \cap \mathcal{S}_2}{\mathcal{S}_1 \cup \mathcal{S}_2}=Jaccard(\mathcal{S}_1,\mathcal{S}_2)
\end{equation}
It generates K random permutations $\pi_1$, $\pi_2$, $\pi_3$, … , $\pi_K$ independently, and the estimator of Minwise Hashing is:
\begin{equation}
R_M(S_1,S_2)=\dfrac{1}{K}\sum_{i=0}^{K}1\{min(\pi_{i}(S_i)=min(\pi_{i}(S_2)))\}
\end{equation}
$R_M$ is an unbiased estimator of $J(S_1, S_2)$, with the variance:
\begin{equation}
Var(R_M)=\dfrac{1}{K}J(S_1, S_2)(1-J(S_1, S_2))
\end{equation}

\subsection{Dynamic Threshold Filter}
\label{filter}

\subsubsection{Problem Description}
In order to estimate the similarity of sets, the hashing algorithms generate K fingerprints (or hash values) by K random permutations, and obtain the unbiased similarity R after the fingerprint comparisons. For example, there are 1 million set pairs and K=1000, we need 1 billion comparisons. These large-scale comparisons spend a lot of computation time and storage space.

\begin{figure}[thb]
\newskip\subfigtoppskip \subfigtopskip = -0.1cm
\centering
\includegraphics[width=.80\linewidth]{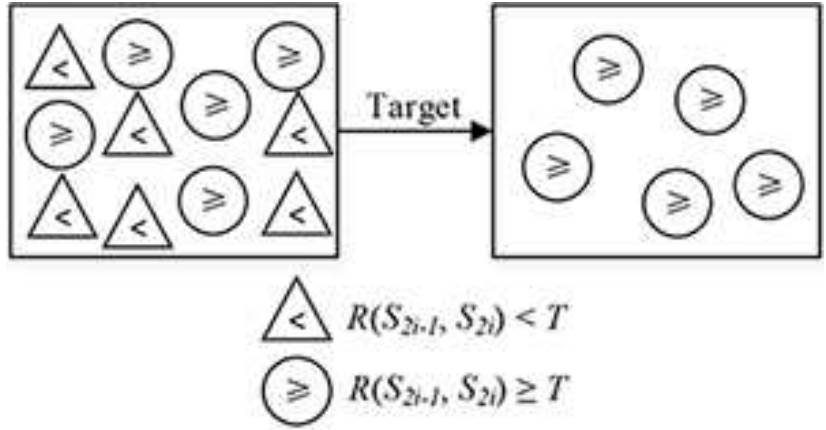}
\vspace{-1mm}
\caption{\small Filter Example }
\label{fig:1}
\end{figure}

According to the clustering algorithm for web pages [10], we also cluster by the visual words in large-scale image datasets and generate a series of image pairs, named $(Image_1, Image_2), (Image_3, Image_4), \dots, (Image_{2n-1}, Image_{2n})$. The corresponding set pairs are defined as $(S_1, S_2), (S_3, S_4), ... , (S_{2n-1}, S_{2n})$. We set a specified threshold T (e.g., 0.5). The target is to find that the set pairs whose estimated similarities are greater than the threshold T after comparisons: $\{(S_{2i-1}, S_{2i}) \mid R(S_{2i-1}, S_{2i}) \geq T, 1\leq i\leq n \}$, as shown in Fig.~\ref{fig:1}.

However, we only care about the pairwise datasets whose similarities are larger than a specified threshold T (e.g., 0.5) in some multimedia applications, e.g. near-duplicate detection, clustering, nearest neighbor search, etc.

The strategy of the filter is: given a small probability e and a specified threshold T at the k-th ($0<k\leq K$) comparison, the set pairs whose similarities are smaller than the threshold $T_{L}(k)$ will be filtered out. Similarly, the set pairs whose similarities are higher than the threshold $T_{U}(k))$ will be output.

According to the above viewpoints, we could set a filter to output or filter out some set pairs in advance, during the process of comparisons. The strategy of the filter is shown in Fig.~\ref{fig:2}. The pre-filtering method greatly reduces the number of comparisons and computation time.

\begin{figure}[thb]
\newskip\subfigtoppskip \subfigtopskip = -0.1cm
\centering
\includegraphics[width=.80\linewidth]{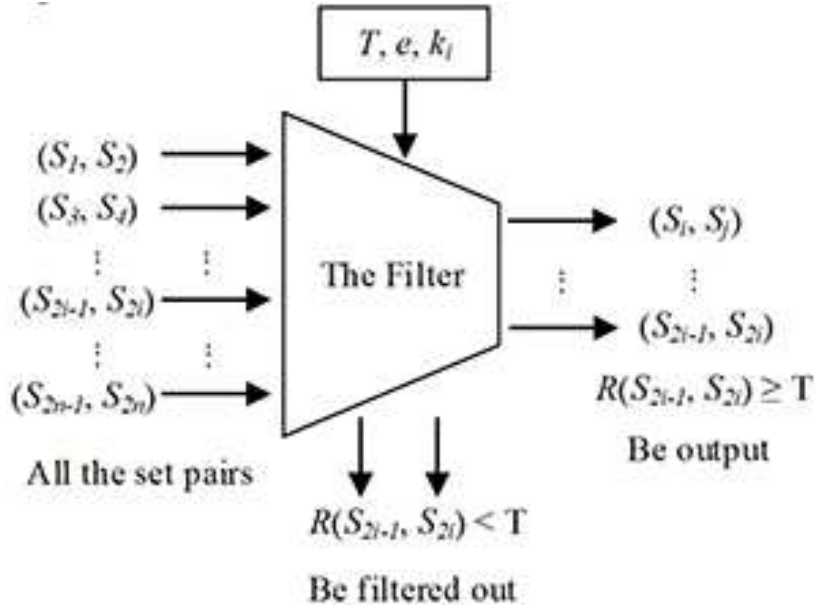}
\vspace{-1mm}
\caption{\small Filter Strategy }
\label{fig:2}
\end{figure}

\subsubsection{Threshold Construction}
When the estimator satisfies the binomial distribution, we can build a dynamic threshold filter through the hypothesis testing and the small probability event. It greatly reduces the calculation time by terminating the unnecessary comparisons in advance. In this pager, we take Minwise Hashing algorithm as an example and build the filter for image similarity measures.

As we know, each image can be represented by a set $S_i$ of visual words through the BoVW model. Then, we could use Minwise Hashing algorithm to deal with the set pair ($S_{2i-1}, S_{2i}$) and obtain the corresponding fingerprints. The relevant definitions are following:

The random variable X is the number of times that the fingerprints are equal, that is:
\begin{equation}
X=\sum_{i=0}^{K}1\{min(\pi_{j}(S_{2i-1})=min(\pi_{j}(S_{2i})))\}
\end{equation}

Obviously, the random variable X satisfies the binomial distribution, denotes as $X \sim B(n,R_M)$. Thus, The distribution function $\mathcal{F}(x)$ of the variable $X$ is denoted as follow:
\begin{equation}
\mathcal{F}(x)=\begin{cases}
            \sum_{i=0}^{m}\tbinom{n}{i}{R_M}^{i}(1-R_M)^{n-i} \ \ \ \ ,  x \leq m\\
            \sum_{i=m+1}^{K}\tbinom{n}{i}{R_M}^{i}(1-R_M)^{n-i}, x > m
           \end{cases}
\end{equation}
where the variable m is in the interval (0, k].

$R_M$ is the estimator of Minwise Hashing after all the K comparisons, there is:

\begin{equation}
R\approx R_M=\dfrac{X}{K}
\end{equation}

The variable $R_{M}(k)$ is defined as the estimated similarity of the Minwise Hashing ,at the k-th ($0<k\leq K$) comparison:
\begin{equation}
R_{M}(k)=\dfrac{X}{k}
\end{equation}

\textbf{The Lower Threshold}
\begin{lemma} \label{lemma:LowThre}
Given a threshold T and a small probability e at the k-th ($0<k\leq K$) comparison, we can obtain the solution $m=m_l$ from the following equation
\begin{equation*}
\label{eqn:9}
\sum_{i=0}^{m}\tbinom{k}{i}{T}^{i}(1-{T})^{k-i}=e, 0<k\leq K
\end{equation*}
Then, we could set the lower threshold $T_{L}(k) = \dfrac{m_l}{k}$. It's obvious that $R_M <T$, when $R_{M}(k)≤T_{L}(k)$.
\end{lemma}

We use the hypothesis testing to prove the Lemma 1.

\begin{proof}
Assume $H_0$:$R_M \geq T$, $H_1$:$R_M < T$, and the random variable X satisfies the binomial distribution: $X \sim B(n,R_M)$, at the k-th ($0<k\leq K$)comparison point. The probability of the event {X$\leq$ m} is:
\begin{equation*}
Pr(X\leq m)=\sum_{i=0}^{m}\tbinom{n}{i}{R_M}^{i}(1-R_M)^{n-i}\leq e
\end{equation*}
where m satisfies $0<m\leq k$ .

Obviously, the event {X$\leq$m} belongs to a small probability event. When $R_{M}(k) \leq T_{L}(k)$, we know that:
\begin{equation*}
\dfrac{X}{k}\leq\dfrac{m_l}{k}
\end{equation*}
Then,
\begin{equation*}
X\leq m_l \leq m
\end{equation*}
In other words, when $R_{M}(k) \leq T_{L}(k)$, the small probability event {X$\leq$m} occurs in an experiment. Therefore, we should reject the hypothesis $H_0$ and accept the hypothesis $H_1$. According to the above discussion, the estimated similarity $R_{M}$ is less than the threshold T and the lemma is proved.
\end{proof}

The following is an example of Lemma~\ref{lemma:LowThre}:

Given a set pair ($S_{2i-1}, S_{2i}$), when K=1000, T=0.5, k=100, we know the random variable X satisfies the binomial distribution, that is, $X \sim B(n,T)$. The distribution function F(x) of X for different m are shown in Table~\ref{tab:prob}.

\begin{table}
\centering
\caption{The values of F(x) for different m} \label{tab:prob}

\centering
\begin{tabular}{cc|cc}

\hline
m& F(x)& m &F(x)\\
\hline
10&	1.53-17&	60&	0.982\\
20&	5.6-10&	70&	0.999\\
30&	3.9-5&	80&	0.999\\
40&	0.028&	90&	0.999\\
50&	0.539&	100&	1\\
\hline
\end{tabular}

\end{table}

When x=20, it is obvious that Pr(X$\leq$20)$\leq$5.6-10 and the event {X$\leq$20} is a small probability event, according to the Table 1.

Assume $H_0$:$R_M \geq T=0.5$, $H_1$:$R_M < T=0.5$, We select a small probability e =5.6-10, m=$m_l$=20 and obtain the lower threshold $T_L$(100)=20/100=0.2. At the 100-th comparison, there is $R_{M}(k)=X\dfrac k\leq T_{L}(k)$=0.2. The probability of the event {X$\leq$20} is only 5.6-10. Clearly, the event {X$\leq$20} is a small probability event. However, it occurs in an experiment. According to the above discussion, we must reject the hypothesis $H_0$: $R_{M}\geq T$=0.5and accept the hypothesis $H_1: R_M<T$=0.5. It means the similarity of the pair ($S_{2i-1},S_{2i}$) is less than threshold T=0.5.

Fig.~\ref{fig:3} explains the original comparison process, the set pairs will be output when their similarities $R_M$ are higher than the threshold T. In our method, we can add the lower threshold $T_L$(100), at the 100-th comparison point, as what is shown in Fig.~\ref{fig:4}. If $R_M(k=100)≤T_L(100)$, we easily obtain the $R_M(k=1000)<T$ and there is no need to the remaining  comparisons. However, when $R_M(k=100)>T_L(100)$, we require the remaining 900 comparisons and calculate the $R_M$(k=1000).

\begin{figure}[thb]
\newskip\subfigtoppskip \subfigtopskip = -0.1cm
\centering
\includegraphics[width=.80\linewidth]{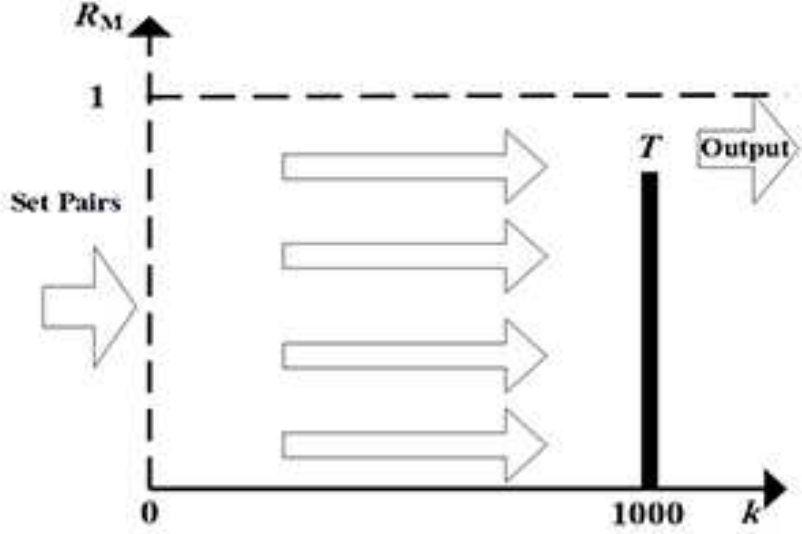}
\vspace{-1mm}
\caption{\small Original comparison process}
\label{fig:3}
\end{figure}

\begin{figure}[thb]
\newskip\subfigtoppskip \subfigtopskip = -0.1cm
\centering
\includegraphics[width=.80\linewidth]{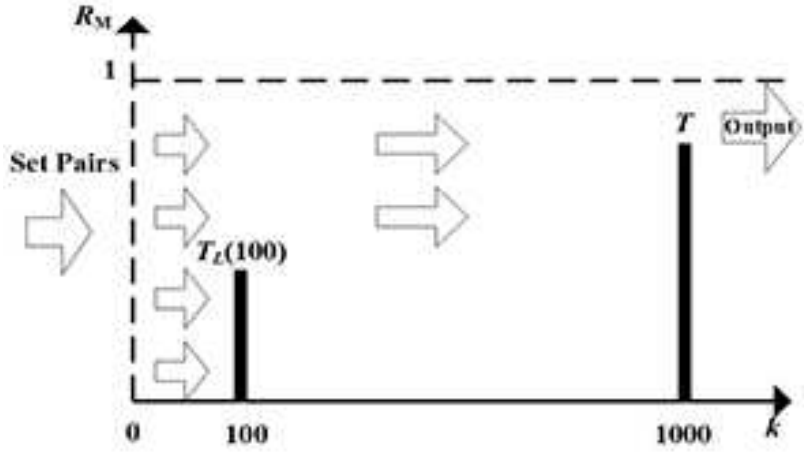}
\vspace{-1mm}
\caption{\small Our comparison process}
\label{fig:4}
\end{figure}

To sum up, during the Minwise Hashing comparison process, we set the small probability e, the threshold T as well as the estimated similarity $R_M$(k) at the k-th ($0<k\leq K$) observation point. If $R_M(k)<T_L(k)$, we can predict $R_M<T$. Therefore, there is no need to the rest of comparisons. Compared the Fig.~\ref{fig:3} with the Fig.~\ref{fig:4}, the method effectively saves the computing time.

\textbf{The Upper Threshold}
Meanwhile, there must be an upper threshold: $T_U$.

\begin{lemma} \label{lemma:UpThre}
Given a threshold T and a small probability e, according to the equation:
\begin{equation*}
\label{eqn:12}
\sum_{i=m+1}^{k}\tbinom{k}{i}{T}^{i}(1-{T})^{k-i}=e, 0<k\leq K
\end{equation*}
We could obtain m=$m_u$ and the upper threshold $T_U(k) = \dfrac{m_u}{k}$. It is clear that $R_M >T$, when $R_M(k)\geq T_U(k)$.
\end{lemma}

\begin{proof}
Assume $H_0$:$R_M < T$, $H_1$:$R_M \geq T$, and the random variable X satisfies the binomial distribution: $X \sim B(n,R_M)$, at the k-th ($0<k\leq K$)comparison. The probability of the event {X$>$ m} is:
\begin{equation*}
Pr(X > m)=\sum_{i=m+1}^{k}\tbinom{k}{i}{R_M}^{i}(1-R_M)^{k-i}\leq e
\end{equation*}
where m satisfies $0<m\leq k$ .

Obviously, the event {X>m} is a small probability event. When $R_{M}(k) > T_{U}(k)$, we know that:
\begin{equation*}
\dfrac{X}{k}>\dfrac{m_u}{k}
\end{equation*}
Then,
\begin{equation*}
X > m_u > m
\end{equation*}
That is to say, when $R_{M}(k) > T_{U}(k)$, the small probability event {X>m} occurs in an experiment, we should refuse the hypothesis $H_0$ and accept the hypothesis $H_1$. Therefore, the Lemma~\ref{lemma:UpThre} is proved.
\end{proof}

The following is an example of Lemma~\ref{lemma:UpThre}.

Given a set pair ($S_{2i-1},S_{2i}$), when K=1000, T=0.5 and k=100, and the random variable X, $X \sim B(n,T)$. The values of 1-F(x) for different m are shown in Table 3.

According to the Table 3, when x=80, it is obvious that $Pr(X>80)<1.35-10$ and the event ${X>80}$ satisfies a small probability event. We assume that: $H_0$:$R_M < T=0.5$, $H_1$:$R_M \geq T=0.5$.
Identically, we select e =1.35-10, m=$m_u$=80 and obtain the upper threshold $T_U$(100)=$\dfrac{80}{100}$=0.8. At the 100-th comparison, there is $R_M$(k)=$\dfrac{X}{k}$>$T_U$(k)=0.8. The probability of the event ${X>80}$ is only 1.35-10. Clearly, the event ${X>80}$ is a small probability event. However, it happens in an experiment. According to the above discussion, we have to reject the hypothesis $H_0:R_M<T$=0.5 and accept the hypothesis $H_1: R_M\geq T$=0.5. That is, the estimated similarity $R_M$ of the pair ($S_{2i-1}, S_{2i}$) satisfies $R_M\geq T$=0.5.

In short, during the Minwise Hashing comparison process, we can set the small probability e, the similarity threshold T as well as the estimated similarity $R_M$(k) at the k-th ($0<k\leq K$) comparison point. It exists the upper threshold $T_U(k)=\dfrac{mu}{k}$. If $R_M(k)>T_U(k)$, we can predict $R_M>T$. There is no need to the following K-k comparisons. As shown in Fig.~\ref{fig:5}, it effectively saves the computing time in image similarity measures. When k=100, we can add the upper threshold $T_U$(100). If $R_M(k=100)>T_U$(100), we can easily draw that $R_M(k=1000)>T$, and there is no need for the following 900 comparisons. If $R_M (k=100)≤T_U$(100), we require the remaining 900 comparisons and continue to calculate the $R_M$(k=1000).

\begin{figure}[thb]
\newskip\subfigtoppskip \subfigtopskip = -0.1cm
\centering
\includegraphics[width=.60\linewidth]{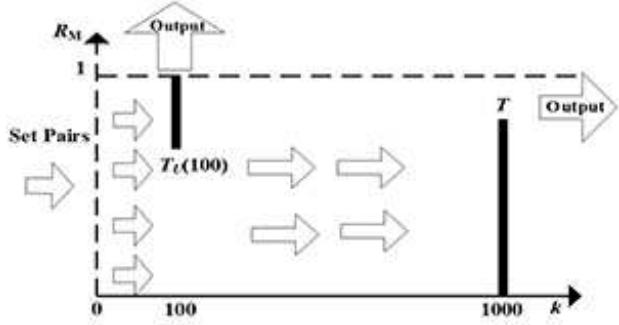}
\vspace{-1mm}
\caption{\small Upper Threshold Filter}
\label{fig:5}
\end{figure}

\subsubsection{The Strategy of the Filter}
Combining the above discussion, we can set the upper and the lower threshold simultaneously at the k-th($0<k\leq K$) comparison point. The threshold filter can eliminate or output the predictable set pairs ahead of time.

Obviously, we can build a series of comparison points {$k_1, k_2, \cdots, k_i, \cdots, k_{n-1}, k_n$}, ($0<k_i\leq K$, 0<i≤n) . That is to say, there are n points in the whole comparison process. Therefore, we can set the lower threshold value $T_L(k_i)$ and the upper threshold $T_U(k_i)$ at each comparison point $k_i$. The dynamic threshold filter of Minhash is shown in the Algorithm 1 and the Fig.~\ref{fig:6}.

\begin{algorithm}[h]
\begin{algorithmic}[1]
\footnotesize
\caption{\bf Dynamic Threshold Filter of Minhash}
\label{alg:dtf}

\INPUT  \\
All the set pairs {$(S_1, S_2),(S_3, S_4), \cdots,(S_{2n-1}, S_{2n})$}; A specified threshold T; A small probability e; A series of comparison points ${k_1, k_2, \cdots, k_i, \cdots, k_{n-1}, k_{n}}, (0<ki\leq K, 0<i\leq n)$
\\

\OUTPUT \\
The set pairs whose estimated similarities are greater than the threshold T:{$(S_{2i-1}, S_{2i})\mid R(S_{2i-1}, S_{2i}) \geq T, 1\leq i≤n $}\\
\FOR{for each set pair ($S_{2i-1}, S_{2i}$)}
    \FOR{k= 1: K}
    \STATE R=R(k);
    \STATE i=1;
    \IF{k=$k_i$}
    \STATE $T_L(k)= \dfrac{m_l}{k}; T_U(k)= \dfrac{m_u}{k}$;
        \IF{$R\geq T_U(k)$}
        \STATE Output the set pair ($S_{2i-1}, S_{2i}$); Break;
        \ENDIF
        \IF{$R\leq T_L(k)$}
        \STATE Filter out the set pair ($S_{2i-1}, S_{2i}$); Break;
        \ENDIF
        \STATE i++;
    \ENDIF
    \ENDFOR
\ENDFOR

\end{algorithmic}
\end{algorithm}

Algorithm 1. The inputs of the algorithm include all the set pairs {$(S_1, S_2), (S_3, S_4), \cdots, \\(S_{2n-1}, S_{2n})$}, as well as the parameters that we set ahead of time: a specified threshold T, a small probability e and a series of comparison points {$k_1, k_2, \cdots, k_i, \cdots, k_{n-1}, k_n$}, ($0<k_i\leq K, 0<i\leq n$). The outputs are the set pairs whose estimated similarities are greater than the threshold T after comparisons: {$(S_{2i-1}, S_{2i}) \mid R(S_{2i-1}, S_{2i}) \geq T, 1\leq i\leq n$ }. Line 3 shows that the algorithm calculates the estimated similarity R=R($k_i$) through the top k fingerprints. Line 5-15 describe that the algorithm calculates the lower threshold $T_L(k_i)$ and the upper threshold $T_U(k_i)$ through the equation (9), (12) at the comparison point k= $k_i$. If $R\geq T_U(k_i)$, we can judge that: $R\geq T$, and the set pair ($S_{2i-1}, S_{2i}$) can be output ahead of time. Besides, if R<$T_L(k_i)$, we could obtain: R<T, and the set pair ($S_{2i-1}, S_{2i}$) can be filtered out in advance. Otherwise, the algorithm should enter the next point $k_i$+1 and continue to compare the remaining fingerprints.

Fig.~\ref{fig:6}. gives an example of the entire comparison process, we can set $k_i$=100, 200, ... , and define the lower threshold value $T_L(k_i)$ and the upper threshold $T_U(k_i)$ for each comparison point $k_i$.

\begin{figure}[thb]
\newskip\subfigtoppskip \subfigtopskip = -0.1cm
\centering
\includegraphics[width=.60\linewidth]{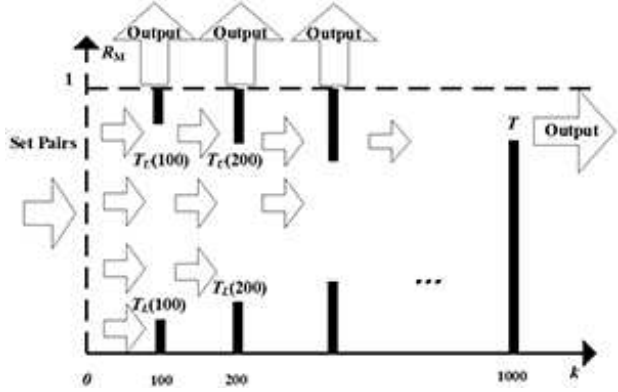}
\vspace{-1mm}
\caption{\small Lower Threshold Filter}
\label{fig:6}
\end{figure}

\subsection{Experiment and Analysis}
\label{exp}
Our empirical studies aim to evaluate the performance of the filter on image dataset “Caltech256”.

Caltech256: contains 30, 607 images of objects, which were obtained from Google image search and from PicSearch.com. All the images were assigned to 257 categories and evaluated by humans in order to ensure image quality and relevance.

We use the Bag-of-Visual-Words model and a 128-dimensional SIFT descriptor for image representations. Therefore, we can obtain a visual set from each image. In this experiment, we only use Minwise Hashing for sampling and generating the fingerprints (or hash values). Of course, we can also use b-Bit Minwise Hashing or One Permutation Hashing. The results and analysis of the experiment are as follows:

\subsubsection{Comparison Time}
(1) We compare the time of comparison between the original Minwise Hash and the dynamic threshold filter. As shown in Fig.~\ref{fig:7}, the Minwise Hashing with a filter greatly reduces the calculation time. When comparing 104 set pairs, we can easily find out that the time of the comparison is inversely proportional to the value of the small probability. The comparison time of the original Minwise Hashing is 30*103 ms. After using a filter with a small probability e=10-3, the calculation time is the least, only 9.3*103 ms, and it is 31\% of the original Minwise Hashing.

\begin{figure}[thb]
\newskip\subfigtoppskip \subfigtopskip = -0.1cm
\centering
\includegraphics[width=.60\linewidth]{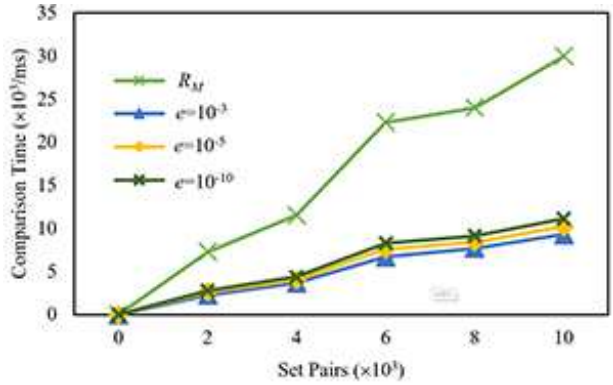}
\vspace{-1mm}
\caption{\small Vary Number of Set Pairs}
\label{fig:7}
\end{figure}

(2) We select three copies of 4000 set pairs and half of their similarities (measured by the Jaccard coefficient) are about 80\%, 50\%, and 30\%, respectively. As shown in Fig.~\ref{fig:8}, the filter is more effective for the set pairs whose similarities are very low or high. If the similarity of the set pairs is mostly high or low, the comparison time will be little.

\begin{figure}[thb]
\newskip\subfigtoppskip \subfigtopskip = -0.1cm
\centering
\includegraphics[width=.60\linewidth]{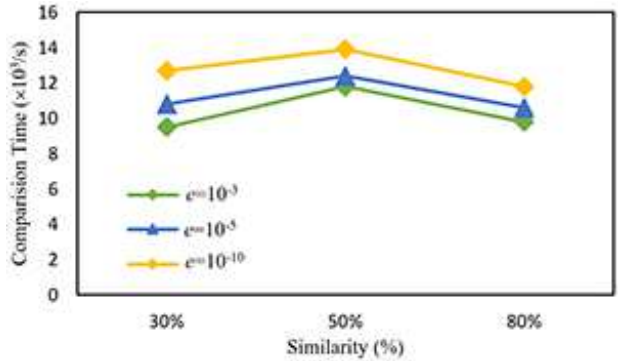}
\vspace{-1mm}
\caption{\small Vary Similarity Threshold}
\label{fig:8}
\end{figure}

\subsubsection{The Filtering Rate}
The Filtering rate refers to the probability that the set pairs are excluded or output in advance, at the k-th comparison point. Therefore, we define the filtering rate (FR) at comparison point k as:
\begin{equation*}\label{eqn:26}
FR(T,k,P_r)=\dfrac{R_M(k)<T_L(k)}{Num}
\end{equation*}

, where the variable Num represents the total number of set pairs and Num=3×105.

\begin{figure}[thb]
\newskip\subfigtoppskip \subfigtopskip = -0.1cm
\centering
\includegraphics[width=.60\linewidth]{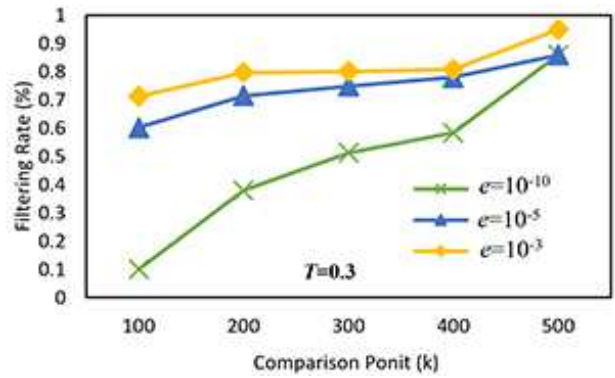}
\vspace{-1mm}
\caption{\small Vary Comparison Point, When T=0.3}
\label{fig:9}
\end{figure}

\begin{figure}[thb]
\newskip\subfigtoppskip \subfigtopskip = -0.1cm
\centering
\includegraphics[width=.60\linewidth]{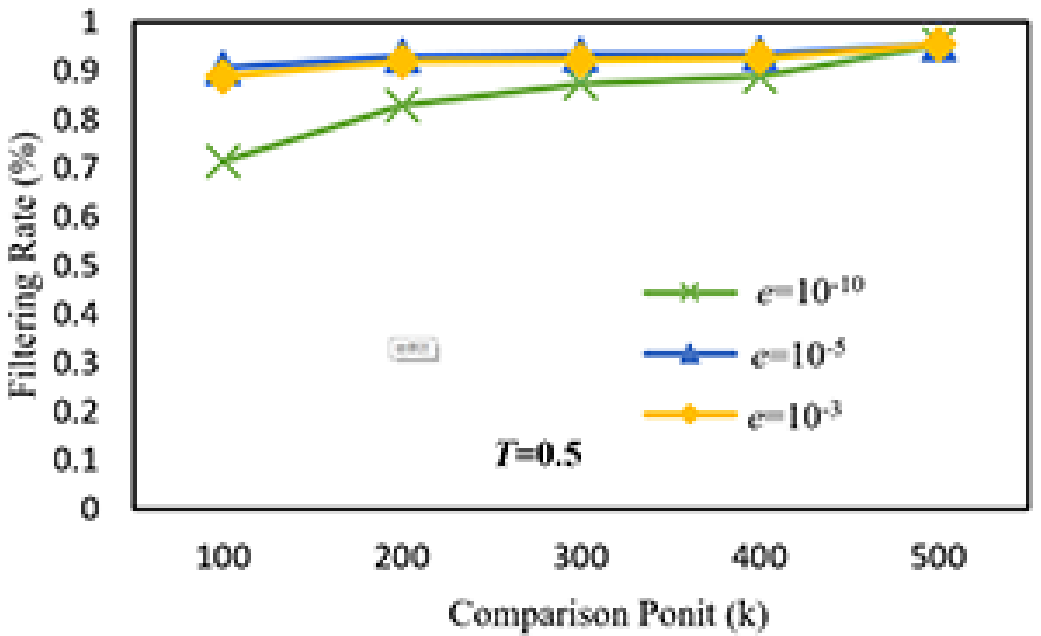}
\vspace{-1mm}
\caption{\small Vary Comparison Point, When T=0.5}
\label{fig:10}
\end{figure}

\begin{figure}[thb]
\newskip\subfigtoppskip \subfigtopskip = -0.1cm
\centering
\includegraphics[width=.60\linewidth]{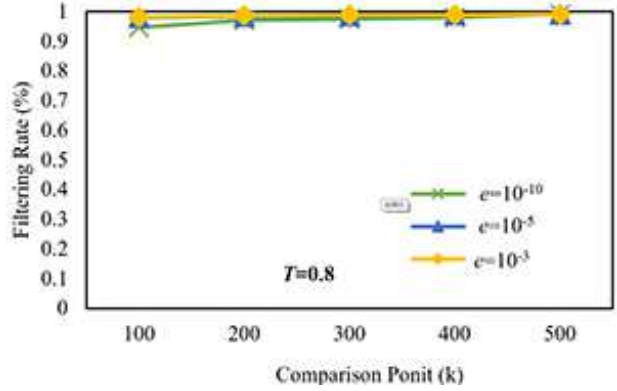}
\vspace{-1mm}
\caption{\small Vary Comparison Point, When T=0.8}
\label{fig:11}
\end{figure}

Obviously, the filtering rate has a great relationship with the input data. When the similarities of the set pairs are mostly low, the filtering rate will be high. According to the above equation, we mainly measure the filtering rate (FR) at different small probability e=10-10, 10-5, 10-3, as shown in Fig. ~\ref{fig:9}, Fig.~\ref{fig:10}, Fig.~\ref{fig:11}. As the small probability increases, the filtering rate is also increasing. It means that the more set pairs are excluded or output in advance, and the less the calculation time is. For example, FR(0.3,100,10-10)=10\%, FR(0.3,100,10-5)=60\%, FR(0.3,100,10-3)=72\% when k = 100, T=0.3. Among them, FR(0.3, 100, 10-3) =72\% means that 73\% of the set pairs save the remaining 900 comparisons.

\subsubsection{The Accuracy of the Filter}
We select three groups of data with actual similarity about 80\%, 50\% and 30\%. Each group include $4*10^{3}$ pairs of sets. In Fig. ~\ref{fig:12}, we mainly analyze the accuracy of the filter. We found that the accuracy of the filter is extremely close to 1.0 and the error can be negligible. That is to say, the image similarity estimated by the filter is almost the same as the original minhash. The main reason may be that the small probability events seldom happen in similarity measurement experiments. In addition, we found that the smaller the probability e is, the higher the accuracy is.

\begin{figure}[thb]
\newskip\subfigtoppskip \subfigtopskip = -0.1cm
\centering
\includegraphics[width=.60\linewidth]{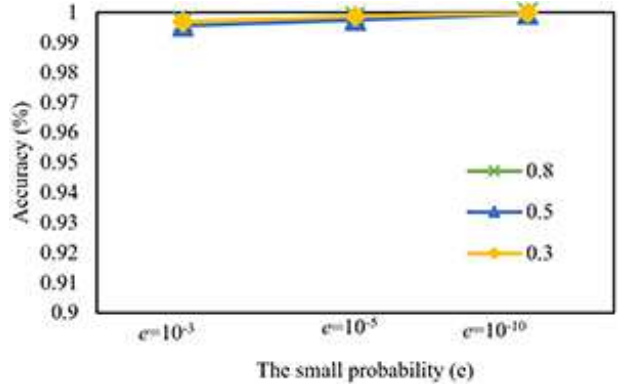}
\vspace{-1mm}
\caption{\small Vary Small Probability}
\label{fig:12}
\end{figure}

\section{Conclusion}
\label{con}
In this paper, we use the Bag-of-Words model and a 128-dimensional SIFT descriptor for image feature representation. The method has achieved great success in computer vision applications, such as image matching, near-duplicate detection and image search. Inspired by the successes of Minwise Hashing in computer vision, we combine binomial distribution with small probability event and propose a dynamic threshold filter for large-scale image similarity measures. It greatly reduces the calculation time by terminating the unnecessary comparison in advance. Besides, we find that the filter can be extended to other hashing algorithms for image similarity measures, such as b-Bit Minwise Hashing, One Permutation Hashing, etc. Our experimental results are based on the image database “Caltech256”, which proves that the filter is effective and correct.

\textbf{Acknowledgments:} This work was supported in part by the National Natural Science Foundation of China
(61379110, 61472450, 61702560), the Key Research Program of Hunan Province(2016JC2018), project (2016JC2011, 2018JJ3691) of Science and Technology Plan of Hunan Province, and Fundamental Research Funds for Central Universities of Central South University (2018zzts588).



\bibliographystyle{spmpsci}      

\bibliography{ref}

\end{document}